\documentclass[runningheads]{llncs}
\usepackage{graphicx}
\usepackage{amssymb}
\usepackage{amsmath}

\allowdisplaybreaks
\begin{document}
\title{Maximising the total weight of on-time jobs on parallel machines subject to a conflict graph
}
\titlerunning{Maximising the total weight of on-time jobs subject to a conflict graph}
%
\author{Yakov Zinder\inst{1}\orcidID{0000-0003-2024-8129} \and
Joanna Berli\'{n}ska\inst{2}\orcidID{0000-0003-2120-2595} \and
Charlie Peter\inst{1}}
\authorrunning{Y. Zinder et al.}
%
\institute{School of Mathematical and Physical Sciences, University of Technology Sydney, Australia \\
\email{Yakov.Zinder@uts.edu.au, 13240039@student.uts.edu.au}
\and
Faculty of Mathematics and Computer Science, Adam Mickiewicz University, Pozna\'n, Poland\\
\email{Joanna.Berlinska@amu.edu.pl} }
\maketitle              
\begin{abstract}
The paper considers scheduling on parallel machines under the constraint that some pairs of jobs cannot be processed concurrently.
Each job has an associated weight, and all jobs have the same deadline.
The objective is to maximise the total weight of on-time jobs.
The problem is known to be strongly NP-hard in general.
A polynomial-time algorithm for scheduling unit execution time jobs on two machines is proposed.
The performance of a broad family of approximation algorithms for scheduling unit execution time jobs on more than two machines is analysed.
For the case of arbitrary job processing times, two integer linear programming formulations are proposed and compared with two formulations known from the earlier literature.
An iterated variable neighborhood search algorithm is also proposed and evaluated by means of computational experiments.

\end{abstract}
\section{Introduction}

The paper is concerned with scheduling a set $N$ of $n$ jobs on $m > 1$ identical parallel machines under the restriction that some pairs of jobs cannot be processed concurrently. The jobs are numbered from 1 to $n$ and are referred to by these numbers, i.e. $N = \{1, ..., n\}$. In order to be completed, a job $j$ should be processed during $p_j$ units of time, where the processing time $p_j$ is integer. Each job can be processed by only one of the machines, and if a machine starts processing a job $j$, then it should process it without interruptions (without preemptions) for the entire job's processing time $p_j$. At any point in time a machine can process at most one job. The only exceptions are the points in time when one job finishes processing and another commences its processing. 

The undirected graph $G(N,E)$, where the set of jobs $N$ is the set of nodes and the set of edges $E$ is the set of all pairs of jobs that cannot be processed concurrently, usually is referred to as the conflict graph. The complement of the conflict graph $G(N, E)$, that is the graph with the same set of nodes $N$ and with the set of edges that is comprised of all edges that are not in $E$, will be referred to as the agreement graph. So, jobs can be processed concurrently only if they induce a complete subgraph (a clique) of the agreement graph.

The processing of jobs commences at time $t = 0$. A schedule $\sigma$ specifies for each $j \in N$ its starting time $S_j(\sigma)$. Since preemptions are not allowed, the completion time of job $j$ in this schedule is 
\[
 C_j(\sigma) = S_j(\sigma)+p_j.
\]
Each job $j$ has an associated positive weight $w_j$ and all jobs have the same deadline $D$. The objective is to find a schedule with the largest total weight of on-time jobs (also referred to as the weighted number of on-time jobs) 
\begin{equation}\label{objF}
 F(\sigma) = \sum_{j \in J} w_j(1-U(C_j(\sigma))),
\end{equation}
where 
\[
 U(t) = \left\{ \begin{array}{ll}
 			1 & \mbox{ if } t > D \\
			0 & \mbox{otherwise}
 		   \end{array}
	  \right.. 
\]

The situation when some jobs cannot be processed concurrently due to technological restrictions and when it may not be possible to complete all jobs during the given time period, arises in planning of maintenance. Thus, an operator of a communication network faces such a situation when it is needed to execute on the parallel computers the so called change requests that modify some parts of this network \cite{ha2019capacitated}. Each computer can execute only one program (change request) at a time, and each change request can be assigned to only one computer. Any two  change requests which affect overlapping parts of the network, cannot be executed concurrently. Every change request is initiated by a technician who remains involved during the entire period of the request's execution, and it may not be possible to execute all change requests during one shift. The change requests have different importance, which is modelled by associating with each change request a certain positive number (weight). The goal is to maximise the total weight of the change requests that are executed during the current shift.

The considered scheduling problem also arises in various  make-to-order systems \cite{arkin1987scheduling}, where $D$ producers are to be assigned to jobs, each of which is a certain production process during the time interval specified by this job.  Each job can be allocated to at most one producer, and each producer can be assigned to any job subject to the following two restrictions: jobs cannot be allocated to the same producer if their time intervals overlap and each producer cannot be assigned to more than $m > 1$ jobs, where $m$ is an integer. Each job has a weight, for example, the associated profit. The goal is to maximise the total weight of all allocated jobs. 
The problem of assigning the producers to jobs in such a make-to-order system is equivalent to the problem of scheduling jobs on $m$ parallel identical machines under the restrictions imposed by a conflict graph. Indeed, associate with each job in the make-to-order system a job which can be processed on any of these $m$ parallel machines and which requires one unit of processing time. In the literature on scheduling, jobs with unit processing (execution) time are referred to as UET jobs. Let $N$ be the set of all UET jobs and let $G(N,E)$ be the conflict graph where the set of edges is the set of all pairs of UET jobs for which the corresponding jobs in the make-to-order system overlap. Then, the two problems are equivalent if all UET jobs have the same deadline $D$ and each UET job has the same weight as the corresponding job in the make-to-order system. Observe that in the make-to-order systems the conflict graphs have a special structure. Such graphs are called interval graphs \cite{balakrishnan2012textbook}.  

The scheduling problems with parallel machines and the restrictions imposed by a conflict graph, are an area of intensive research. The main focus in this research was on  the minimisation of the makespan  
\begin{equation}\label{makespan}
 C_{max}(\sigma) = \max_{j \in N} C_j(\sigma),
\end{equation}
and despite of various applications and the challenging mathematical nature, the maximisation of (\ref{objF}) has attracted much less attention than it deserves. The paper addresses this gap in the knowledge by establishing the existence of a polynomial-time algorithm, analysing the performance of approximation algorithms, and by presenting integer linear programming formulations and heuristics together with their comparison by means of computational experiments. 

\section{Related work}

\label{motivation}

As has been mentioned above, the majority of publications on scheduling on parallel identical machines under the restrictions imposed by a conflict graph, pertain to the makespan minimisation \cite{baker1996mutual,jansen2003mutual,gardi2008mutual,gardi2009mutual,demange2007time,even2009scheduling,bendraouche2012scheduling,bampis2014bounded,bendraouche2015scheduling,bendraouche2016scheduling,kowalczyk2017exact,mohabeddine2019new}. 
The publications on the scheduling problems with UET jobs constitute considerable part of the generated literature, including, in particular, \cite{baker1996mutual,jansen2003mutual,gardi2008mutual,gardi2009mutual,even2009scheduling} from the list above. As has been shown in \cite{baker1996mutual}, the problem of scheduling UET jobs on parallel identical machines with a conflict graph and the objective function (\ref{makespan}) arises in balancing the load on parallel processors when partial differential equations are solved using the domain decomposition. In this application, each region of the domain is viewed as a job and each pair of regions which have common points is viewed as an edge in the conflict graph. Another application of the UET case of the makespan minimisation, mentioned in \cite{baker1996mutual}, is the exams timetabling problem, where two exams cannot be scheduled concurrently if some students must sit for both of them.

The makespan minimisation problem with parallel identical machines, UET jobs and a conflict graph closely relates to two problems that are among the central in the graph theory: the graph coloring problem \cite{diestel2012graph} and the maximum matching problem \cite{lovasz2009matching}. In the graph coloring problem, each color class is the set of jobs that are processed concurrently, and therefore, this problem has an additional restriction that the cardinality of each color class cannot exceed the limit imposed by the number of machines. Since the problem of partitioning a graph into triangles is NP-complete in the strong sense \cite{garey1979computers}, the graph coloring problem (and the equivalent makespan minimisation problem) is NP-hard in the strong sense even in the case when each color class cannot have more than three nodes. Therefore, the research on the graph coloring with a limit on the size of each color class was focused on various particular classes of graphs \cite{bodlaender1995restrictions}. 
The relevance of the maximum matching problem follows from the observation that any two jobs, which can be processed concurrently, correspond to an edge in the agreement graph and the minimisation of the makespan for $m=2$ is equivalent to finding the largest number of edges in the agreement graph that do not have common nodes  \cite{baker1996mutual}. Observe that the maximum matching problem as well as the maximum weight matching problem can be solved in polynomial time \cite{edmonds1965maximum,edmonds1965paths}.

If all jobs can be completed by the deadline $D$, then the corresponding schedule is optimal for (\ref{objF}). Therefore, NP-hardness results for the makespan imply the NP-hardness for the total weight of on-time jobs. Furthermore, the well known result of \cite{zuckerman2006linear} established that, for all $\varepsilon > 0$, it is NP-hard to find an approximation for the maximum clique problem even within $n^{1-\varepsilon}$. The latter implies the inapproximability of the maximisation of (\ref{objF}) in the case when $m = n$, $D=1$, and the jobs are UET jobs with equal weights.

\section{Scheduling UET jobs}

This section assumes that the optimal makespan is greater than $D$, because otherwise the maximisation of (\ref{objF}) is equivalent to the makespan minimisation.

\subsection{Polynomial-time algorithm for $m=2$}

As has been mentioned above, if $m=2$, a schedule that minimises the makespan can be found by constructing a maximum matching in the agreement graph. 
In contrast, it can be shown that a maximum weight matching may contain no edges which represent pairs of jobs executed concurrently in any schedule that maximises (\ref{objF}). As will be shown below, this scheduling problem can be solved in polynomial time by using the idea suggested in \cite{plesnik1999constrained} for the generalisation of the maximum weight matching called in  
\cite{plesnik1999constrained} the constrained matching.

Let $G(N,E^c)$ be the complement of the conflict graph $G(N,E)$. This agreement graph is transformed in two steps. According to the first step, each edge $e = \{j, g\}$ in $E^c$ is assigned the weight $u_e = w_j + w_g$ and, for each job $j \in N$, a new node $j'$ and the edge $e = \{j, j'\}$ with the weight $u_e = w_j$ are introduced. This doubles the number of nodes and increases the number of edges by $n$. Denote the set of new nodes by $N'$ and the set of new edges by $E'$. The second step is based on the idea in \cite{plesnik1999constrained}:  the introduction of the set $Q$ of $2(n - D)$ additional nodes together with the set $E_Q$ of $2(n - D)|N\cup N'|$ edges, all of the same weight $u = 0$, that link each node in $Q$ with every node in $N \cup N'$. 

Any optimal schedule $\sigma$ induces a perfect matching $M$ in $G(N\cup N'\cup Q, E^c\cup E'\cup E_Q)$, where each pair of jobs $j$ and $g$, processed concurrently and completed on time, induces the edge $\{j, g\}$; each job $j$, which is completed on time and is not processed concurrently with any other job, induces the edge in $E'$ that covers this job; and the remaining $2(n-D)$ edges are from $E_Q$, each covering a distinct node in $N\cup N'$, which is not covered by the induced edges from $E^c\cup E'$, with a distinct node in $Q$.

Conversely, any maximum weight perfect matching $M$ in $G(N\cup N'\cup Q, E^c\cup E'\cup E_Q)$ induces a schedule $\sigma$, where
two jobs $j$ and $g$ are processed in $\sigma$ concurrently only if $\{j, g\} \in M$; the jobs, covered by the edges in $E_Q$, are tardy (being perfect, $M$ must contain $2(n-D)$ edges from $E_Q$) and the remaining jobs, covered by the remaining $D$ edges (all these edges are from $E^c\cup E'$) are completed on time. 

In the discussion above, the weight of the matching $M$ induced by $\sigma$, is equal to $F(\sigma)$, and $F(\sigma)$ for the schedule $\sigma$ induced by the matching $M$, is equal to the weight of $M$. Therefore, any maximum weight perfect matching induces an optimal schedule. Since such a matching can be constructed in polynomial time \cite{lovasz2009matching}, an optimal schedule is also constructed in polynomial time.

\subsection{Approximation algorithms for $m > 2$}

This subsection is concerned with the performance of the algorithms that construct a schedule for arbitrary $m$, using as a subroutine  the method of constructing an optimal schedule for $m=2$. For any algorithm $A$, denote by $\sigma^A$ a schedule constructed according to this algorithm. For any schedule $\sigma$, let $J(\sigma)$ be the set of jobs that are completed on time in this schedule. Denote by $\sigma^*$ an optimal schedule, i.e. a schedule with the largest value of (\ref{objF}), and by $\sigma_2$ a schedule with the largest total weight of on-time jobs for the problem, obtained from the original problem by replacing the original number of machines $m$ by 2. Observe that $\sigma_2$ can be constructed in polynomial time, using the method described in the previous subsection. Let $\mathfrak{A}$ be the set of all algorithms $A$ such that  
$
 J(\sigma_2) \subseteq J(\sigma^A).
$
\begin{theorem}
For each $A \in \mathfrak{A}$, 
\[
 F(\sigma^*) \le \frac{m}{2}\;F(\sigma^A)  
\]
and this performance guarantee is tight.
\end{theorem}
\begin{proof}
For each integer $1 \le t \le D$, denote by $k(t)$ the number of jobs processed in schedule $\sigma^*$  in the time slot $[t-1, t]$, i.e. the number of jobs $j$ such that $C_j(\sigma^*)=t$. By virtue of the assumption that the makespan is greater than $D$, for all considered integers, $k(t) \ge 1$. Let $j_{t,1}$, ..., $j_{t,k(t)}$ be the jobs, processed in the time slot $[t-1, t]$ and numbered in a nonincreasing order of their weights, i.e.   
\begin{equation}\label{theorem:w>w}
 w_{t, 1} \ge ... \ge w_{t,k(t)},
\end{equation}
and denote 
\[
 u(t) =    \left\{
            \begin{array}{ll}
             0 & \mbox{ if } k(t) \le 2\\
             w_{t,3} & \mbox{otherwise}
            \end{array}
        \right..
\]

Let $\sigma^*_2$ be any schedule for the agreement graph $G(N,E^c)$ and two machines such that, for any integer $1 \le t \le D$ and each integer $1 \le i \le \min[2, k(t)]$, $C_{j_{t, i}}(\sigma^*_2) = t$. Then, taking into account (\ref{theorem:w>w}), for each integer $1 \le t \le D$,
\[
 u(t) \le \frac{1}{2}\sum_{\{j:\; C_j(\sigma^*_2)=t\}}w_j
\]
and consequently, using $J(\sigma_2) \subseteq J(\sigma^A)$,
\[
 F(\sigma^*) =
 \sum_{1 \le t \le D}\sum_{\{j:\; C_j(\sigma^*)=t\}}w_j \le \sum_{1 \le t \le D}\left [\sum_{\{j:\; C_j(\sigma^*_2)=t\}}w_j + (m-2)u(t) \right ]
\]
\[
 \le \frac{m}{2}\sum_{1 \le t \le D}\sum_{\{j:\; C_j(\sigma^*_2)=t\}}w_j \le \frac{m}{2}F(\sigma_2)  \le \frac{m}{2}\;F(\sigma^A).
\]

The performance guarantee $\frac{m}{2}$ is tight for each algorithm $A \in \mathfrak{A}$, because, as will be shown below, for any  $\varepsilon>0$, there exists an instance for which 
\begin{equation}\label{theorem: F}
 F(\sigma^*) > \left( \frac{m}{2} - \varepsilon \right) F(\sigma^A).
\end{equation}
Indeed, assume that $\varepsilon \le \displaystyle \frac{1}{2}$ and consider the instance where the agreement graph is comprised of a complete graph of $mD$ nodes and $D$ disjoint edges. All jobs, constituting the complete graph, have the same weight $w$, whereas all jobs, constituting the disjoint edges, have the same weight $w+\delta$, where 
\begin{equation}\label{theorem: delta}
 0 < \delta < \frac{2\varepsilon w}{m-2\varepsilon}.
\end{equation}
Since each job, covered by the $D$ disjoint edges, has a weight greater than the weight of any job that is not covered by these $D$ edges, the set $J(\sigma_2)$ is the set of all jobs covered by the $D$ disjoint edges, and 
\[
 F(\sigma_2) = 2D(w+\delta).
\]
On the other hand, any job covered by one of the $D$ disjoint edges, can be processed concurrently only with one job, namely the job covered by the same edge. Furthermore, since, for all $m \ge 3$, 
\[
 \frac{2\varepsilon w}{m-2\varepsilon} \le \frac{(m-2)w}{2},
\]
by virtue of (\ref{theorem: delta}), 
\[
 mw >2 (w+\delta),
\]
and therefore, $F(\sigma^*) = mwD$. Consequently, taking into account (\ref{theorem: delta}),
\[
 \frac{F(\sigma^*)}{F(\sigma^A)} = \frac{mwD}{2D(w+\delta)} = \frac{m(w+\delta)}{2(w+\delta)} - \frac{m\delta}{2(w+\delta)}
 > \frac{m}{2} - \varepsilon,
\]
which implies (\ref{theorem: F}).
\qed
\end{proof}

\section{Scheduling jobs with arbitrary processing times}
\label{sec:algo}

\subsection{Integer linear programming formulations}
\label{sec:ILP}

Two integer linear programming formulations for the case of arbitrary processing times were proposed in \cite{ha2019capacitated}.
This section presents two new formulations, based on the approach from \cite{sevaux_2001}.
Let $M=\{1,\dots,m\}$ be the set of available machines. 
Note that the maximum position on which an on-time job may be scheduled on a machine is at most $k_{max}=\min\{n, \lfloor D/\min_{j=1}^n\{p_j\} \rfloor \}$.
Let $K=\{1,\dots,k_{max}\}$ be the set of available positions of on-time jobs.
Let $e$ denote the number of conflicting pairs of jobs.

For the first formulation, denoted by ILP1, define for all $j\in N$, $k\in K$ and $l\in M$ binary variables $u_{jkl}$ such that $u_{jkl}=1$ if job $j$ is not tardy and is scheduled at position $k$ on machine $l$, and $u_{jkl}=0$ otherwise.
For any $j\in N$, let $U_j=1$ if job $j$ is tardy, and $U_j=0$ if it is completed on time.
Moreover, for any two conflicting jobs $j,g\in N$, let $y_{jg}=0$ if job $j$ precedes job $g$, i.e., job $j$ finishes before $g$ starts, and $y_{jg}=1$ otherwise.
Finally, let $t_{kl}$ denote the starting time of the job at position $k$ on machine $l$, and let $\tau_j$ be the starting time of job $j$.
The considered problem can be stated as follows.

\begin{eqnarray}
\textrm{(ILP1)}&&\textrm{minimise}\quad  \sum_{j=1}^n w_jU_j \label{eq-ilp1-obj}\\
\textrm{s.t.}&&\sum_{k=1}^n \sum_{l=1}^m u_{jkl} + U_j =1  \quad \forall \ j\in N \label{eq-ilp1-all-jobs}\\
&&\sum_{j=1}^n u_{jkl} \leq 1 \quad  \forall\ k\in K, \ \ l\in M \label{eq-ilp1-pos-once}\\
&&t_{kl} + \sum_{j=1}^n p_j u_{jkl} \leq t_{k+1,l} \quad  \forall\ k\in K \setminus \{k_{max}\},  l\in M \label{eq-ilp1-no-overlap}\\
&&t_{kl}+ \sum_{j=1}^n p_j u_{jkl}  \leq D \quad  \forall\ k\in K,  l\in M \label{eq-ilp1-not-late}\\
&&\tau_j + D(1-u_{jkl}) \geq t_{kl}  \quad  \forall\ j \in N,  k\in K,   l\in M \label{eq-ilp1-times-1}\\
&&t_{kl} + D(1-u_{jkl}) \geq \tau_j  \quad  \forall\ j\in N,  k\in K,  l\in M \label{eq-ilp1-times-2}\\
&& \tau_j+p_j(1-U_j) - Dy_{jg} \leq   \tau_g  \quad  \forall\ (j,g)\in E \label{eq-ilp1-j-before-g}\\
&& y_{jg}+y_{gj}\leq 1  \quad \forall\ (j,g)\in E,  j<g \label{eq-ilp1-no-conflicts}\\
&&t_{kl}\geq 0  \quad  \forall\ k\in K,  l\in M \label{eq-ilp1-t-type}\\
&&\tau_j \geq 0  \quad  \forall\ j\in N \label{eq-ilp1-tau-type}\\
&&U_j \in \{0,1\}  \quad  \forall\ j\in N \label{eq-ilp1-Uj-type}\\
&&u_{jkl} \in \{0,1\} \quad  \forall\ j\in N,\ \ k\in K, \ \ l\in M  \label{eq-ilp1-ujkl-type}\\
&&y_{jg} \in \{0,1\}  \quad  \forall\ (j,g)\in E \label{eq-ilp1-y-type}
\end{eqnarray}

The objective in the above program is to minimise the weighted number of tardy jobs (\ref{eq-ilp1-obj}), which is equivalent to maximising the weighted number of on-time jobs.
Equations (\ref{eq-ilp1-all-jobs}) guarantee that each job is either tardy or scheduled at exactly one position on one machine.
By (\ref{eq-ilp1-pos-once}), at most one job is scheduled at each position on each machine.
Constraints (\ref{eq-ilp1-no-overlap}) ensure that no two jobs are executed at the same time on the same machine.
According to (\ref{eq-ilp1-not-late}), a job that is scheduled at position $k$ on machine $l$ is completed by time $D$.
Inequalities (\ref{eq-ilp1-times-1}) and (\ref{eq-ilp1-times-2}) ensure that if job $j$ is scheduled at position $k$ on machine $l$, then $\tau_j=t_{kl}$.
Constraints (\ref{eq-ilp1-j-before-g}) guarantee that for any two conflicting jobs $j$ and $g$, if $j$ is scheduled and $y_{jg}=0$, then $j$ is finished before job $g$ starts.
By (\ref{eq-ilp1-no-conflicts}), no two conflicting jobs are executed at the same time.

To construct the second formulation, denoted by ILP2, let binary variables $U_j$ and $y_{jg}$ be defined as in ILP1.
Additionally, for each $j\in N$ and $t\in \{0,\dots, D-p_j\}$, let $v_{jt}=1$ if job $j$ starts at time $t$ (on an arbitrary machine), and $v_{jt}=0$ otherwise.
The optimal schedule can be found in the following way.

\begin{eqnarray}
\textrm{(ILP2)}&& \textrm{minimise}\quad  \sum_{j=1}^n w_jU_j \label{eq-ilp2-obj}\\
\textrm{s.t.}&&\sum_{t=0}^{D-p_j} v_{jt} + U_j=1  \quad \forall\ j\in N \label{eq-ilp2-all-jobs-once}\\
&&\sum_{j=1}^n \sum_{s=\max\{0,t-p_j+1\}}^{\min\{t,D-p_j\}} v_{js} \leq m \quad \forall\ t=0,\dots,D-1 \label{eq-ilp2-machines}\\
&& \sum_{t=0}^{D-p_j} t v_{jt} + p_j(1-U_j)  - Dy_{jg} \leq   \sum_{t=0}^{D-p_g}t v_{gt}  \quad \forall\ (j,g)\in E \label{eq-ilp2-j-before-g}\\
&& y_{jg}+y_{gj}\leq 1  \quad\forall\ (j,g)\in E,  j<g \label{eq-ilp2-no-conflicts}\\
&&v_{jt} \in \{0,1\}  \quad \forall\ j\in N,  t=0,\dots,D-p_j \label{eq-ilp2-v-type}\\
&&U_j \in \{0,1\}  \quad \forall\ j\in N \label{eq-ilp2-Uj-type}
\end{eqnarray}

Once again, the weighted number of tardy jobs is minimised (\ref{eq-ilp2-obj}).
Constraints (\ref{eq-ilp2-all-jobs-once}) guarantee that each job $j$ is either tardy or scheduled to start at exactly one moment $t\leq D-p_j$.
At most $m$ jobs are executed at any time $t$ by (\ref{eq-ilp2-machines}).
Inequalities (\ref{eq-ilp2-j-before-g}) ensure that for any two conflicting jobs $j$ and $g$, if $j$ is scheduled and $y_{jg}=0$, then $j$ is finished before job $g$ starts.
Conflicting jobs are not executed at the same time by (\ref{eq-ilp2-no-conflicts}).

\subsection{Heuristic algorithms}
\label{sec:heur}

This section presents heuristic algorithms for the analysed problem.
In the schedule representation used, a list of assigned jobs and their starting times is stored for each machine.
Additionally, a separate list of tardy jobs is maintained.

Firstly, a variable neighborhood search algorithm VNS is proposed.
Variable neighborhood search is a metaheuristic for solving combinatorial optimisation problems, introduced by \cite{mladenovic1997VNS}.
It consists in a systematic change of neighborhood within a local search algorithm.
There exist many variants of variable neighborhood search, which have been successfully applied in many areas \cite{gendreau2010metaheuristics}.
In this work, the variable neighborhood descent method is used.
Consider $k_{max}$ neighborhoods $N_1,\dots,N_{k_{max}}$.
Let $x$ be the initial solution passed to the algorithm.
At the beginning, the current neighborhood number is $k=1$.
In each iteration of the algorithm, the best solution $x'$ in neighborhood $N_k(x)$ is found.
If an improvement is obtained, i.e. $x'$ is better than $x$, then $x$ is updated to $x'$ and $k$ is set to 1.
Otherwise, $k$ is increased by 1 and the next neighborhood $N_k(x)$ is considered.
The process continues until $k$ exceeds $k_{max}$.

In the proposed variable neighborhood search, the initial schedule is delivered by a list scheduling algorithm which processes the jobs in the weighted shortest processing time order.
Thus, jobs are first sorted in such a way that $p_1/w_1\leq \dots \leq p_n/w_n$.
Afterwards, each consecutive job $j$ is assigned to the machine that is the first to finish processing assigned jobs in the current schedule.
The earliest possible starting time of job $j$ is then computed, taking into account the time when processing on the selected machine finishes, as well as the processing intervals of already scheduled jobs conflicting with $j$.
If it is not possible to finish job $j$ by time $D$, this job is removed from the machine's list and added to the list of tardy jobs. 
After the initial schedule is constructed, the following six neighborhoods are considered in the search procedure.
\begin{itemize}
\item $N_1(\sigma)$ consists of all schedules obtained from $\sigma$ by swapping a single pair of jobs on one machine;
\item $N_2(\sigma)$ consists of all schedules obtained from $\sigma$ by moving one job to a different position on the machine it is assigned to;
\item $N_3(\sigma)$ consists of all schedules obtained from $\sigma$ by swapping a pair of jobs assigned to different machines;
\item $N_4(\sigma)$ consists of all schedules obtained from $\sigma$ by moving a job scheduled on any machine $i$ to an arbitrary position on a different machine;
\item $N_5(\sigma)$ consists of all schedules obtained from $\sigma$ by replacing an arbitrary scheduled job by an arbitrary tardy job;
\item $N_6(\sigma)$ consists of all schedules obtained from $\sigma$ by moving an arbitrary tardy job to an arbitrary position on any machine.
\end{itemize}
The order of the neighborhoods was selected on the basis of preliminary computational experiments, which suggested that the best results are produced when the neighborhoods obtained by changes on a single machine are considered first, and the neighborhoods obtained by assigning to machines the jobs from the tardy list are considered last.
Naturally, after changing the assignment of jobs to machines and positions, the starting times of the jobs have to be recomputed, taking into account job conflicts.
If due to the changes made, a job that was assigned to a machine becomes tardy, it is removed from the machine's list of jobs and added to the list of tardy jobs.
Moreover, after executing the changes, the list of tardy jobs is scanned in the weighted shortest processing time order, and additional jobs from this list are assigned to the least loaded machines, if possible.

Secondly, the proposed VNS algorithm is embedded in the iterated search framework.
The iterated variable neighborhood search algorithm IVNS starts with executing the VNS.
The obtained schedule is then modified by making $\lceil 0.05 n \rceil$ changes consisting in moving a random (scheduled or tardy) job to a random position on a random machine.
Job starting times are recomputed, and additional jobs are scheduled if possible, as explained in the description of the VNS algorithm.
After this shaking step, the variable neighborhood search is run again.
The whole procedure is repeated $r=10$ times.
The number of changes made in the shaking procedure was selected on the basis of preliminary computational experiments.
The number of iterations $r$ was chosen as a compromise between solution quality and the running time of the algorithm.

\section{Computational experiments}
\label{sec:exper}

In this section, the results of computational experiments on the quality of the proposed algorithms are presented.
The algorithms were implemented in C++, and integer linear programs were solved using Gurobi.
In addition to the algorithms presented in Section \ref{sec:algo}, linear programming formulations F1 and F2 proposed in \cite{ha2019capacitated} were implemented.
The experiments were performed on an Intel Core i7-7820HK CPU @ 2.90GHz with 32GB RAM.

In the generated test instances, the number of machines was $m\in [2,10]$, and the number of jobs was $n\in \{5m, 10m\}$.
The job processing times $p_j$ were chosen randomly from the interval $[50,150]$, and the job weights $w_j$ were selected randomly from the range $[1,5]$.
A parameter $\delta>0$ was used to control the ratio between the available time window and the expected amount of work per machine.
Since the expected duration of a job was 100, the common deadline of all jobs was set to $D= 100 \delta n/m$.
The number of conflicting jobs was controlled by a parameter $c\in (0,1)$.
For a given value of $c$, the conflict graph $G$ contained $c \binom{n}{2}$ randomly chosen edges.
In the experiments presented in this paper, $c=0.1$ was used.
For each analysed combination of instance parameters, 30 tests were generated and solved.

Many instances could not be solved to optimality in reasonable time using the integer linear programs.
Therefore, a one hour time limit was imposed on algorithms ILP1, ILP2, F1 and F2.
Since the optimum solutions were not known in all cases, the total weights of on-time jobs returned by the respective algorithms were compared to an upper bound $UB$ defined as the smallest upper bound obtained by Gurobi while solving the integer linear programs.
Schedule quality was measured by the relative percentage error with respect to $UB$.

\begin{table}
\caption{Performance of the integer linear programs.
\label{tab-ILP}
}
\begin{center}
\begingroup
\setlength{\tabcolsep}{5pt}
\begin{tabular}{lrrrrrr}
\hline
Algorithm & \multicolumn{2}{c}{$m=2, n=10$} &  \multicolumn{2}{c}{$m=2,n=20$} & \multicolumn{2}{c}{$m=4,n=40$}\\
 & error (\%) & time (s) & error (\%)  & time (s) & error (\%) & time (s)\\
\hline
ILP1 & 0.00 & 4.22E$-2$ & 0.00 & 2.12E$-1$ & 0.00 & 1.50E$+2$\\
ILP2 & 0.00 & 2.36E$+0$ & 0.00 & 1.55E$+1$ & 0.00 & 2.45E$+2$\\
F1 \cite{ha2019capacitated} & 0.00 & 1.71E$+2$ & 0.00 & 3.61E$+3$ & 1.09 & 3.61E$+3$\\
F2 \cite{ha2019capacitated} & 0.00 & 2.14E$+1$ & 0.00 & 3.61E$+3$ & 1.68 & 3.61E$+3$\\
\hline
\end{tabular}
\endgroup
\end{center}
\end{table}

In the first experiment, the performance of integer linear programs ILP1, ILP2, F1 and F2 was compared on instances of various sizes with $\delta=0.7$.
The obtained results are presented in Table \ref{tab-ILP}.
All analysed algorithms were able to solve all the instances with $m=2$ and $n=10$ within the imposed time limit.
Still, the running times of the algorithms differed significantly.
ILP1 had the best average running time, followed by ILP2, then F2, and finally F1, which was four orders of magnitude slower than ILP1.
The optimum schedules were also found by all algorithms for all instances with $m=2$ and $n=20$.
However, F1 and F2 did not finish their computations within an hour for any such tests.
Although the optimum solutions were found, the upper bounds computed during one hour using these formulations were larger.
Note that while the running time of the integer linear programs was limited to one hour, the total algorithm running time also includes building the model and retrieving the solution found.
This is why the average times reported for F1 and F2 exceed one hour.
All generated instances with $m=4$, $n=40$ were solved to optimality by ILP1 and ILP2, but F1 and F2 did not find the best schedules for some of these tests.
Their average errors were 1.09\% and 1.68\%, correspondingly.
It can be concluded that the formulations ILP1 and ILP2 are more efficient than F1 and F2 proposed in \cite{ha2019capacitated}. 
As the running times of F1 and F2 were very long even for small instances, these two algorithms were not used in the remaining experiments.

\begin{figure}[t]
\unitlength1cm
\begin{picture}(11,5)
\put(2.9,0){a)}
\put(0.1,0.5){\includegraphics[height=4.2cm]{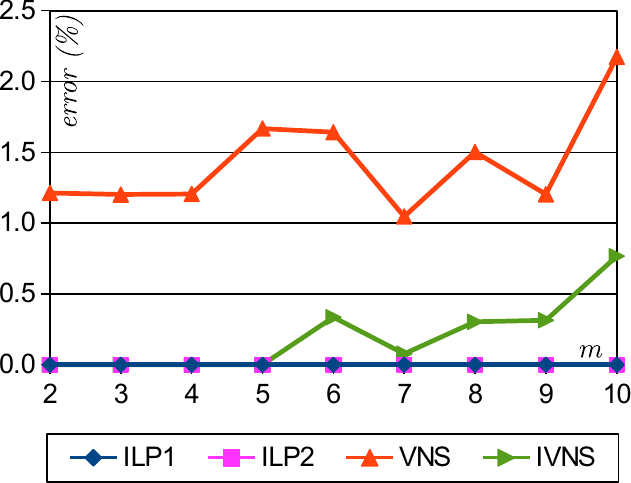}}
\put(9.5,0){b)}
\put(6.5,0.5){\includegraphics[height=4.2cm]{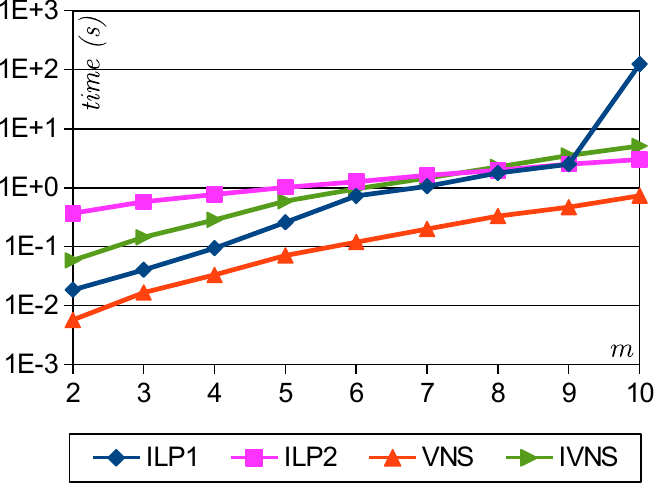}}
\end{picture}
\caption{\label{fig-vs-m-n-5m-delta-0.3}
Algorithm performance vs. $m$ for $n=5m$, $\delta=0.3$.
a) Average quality,
b) average execution time.
}
\end{figure}

Fig. \ref{fig-vs-m-n-5m-delta-0.3} presents the results obtained for instances with $m=2,\dots,10$, $n=5m$ and $\delta=0.3$.
All tests in this group were solved to optimality by both ILP1 and ILP2 (see Fig. \ref{fig-vs-m-n-5m-delta-0.3}a).
VNS produced good schedules, its average error is below 2\% for all $m\leq 9$, and only reaches 2.17\% for $m=10$.
Using IVNS leads to obtaining substantially better results, although at a higher computational cost.
All tests with $m\leq 5$ were solved to optimality by this algorithm, and the largest average error, obtained for $m=10$, is only 0.77\%.
Naturally, the running times of all algorithms increase with growing $m$ (cf. Fig. \ref{fig-vs-m-n-5m-delta-0.3}b).
VNS is the fastest among the analysed algorithms.
ILP1 is faster than IVNS and ILP2 for $m\leq 9$, which means that it is very suitable for solving small instances with a small $\delta$.
However, its average running time rapidly increases when $m=10$.
This is caused by the fact that ILP1 reached the time limit of 1 hour for one instance with $m=10$.
Contrarily, ILP2 is slower than IVNS for $m\leq 7$, but faster than IVNS for $m\geq 8$.

\begin{figure}[t]
\unitlength1cm
\begin{picture}(11,5)
\put(2.7,0){a)}
\put(0,0.5){\includegraphics[height=4.2cm]{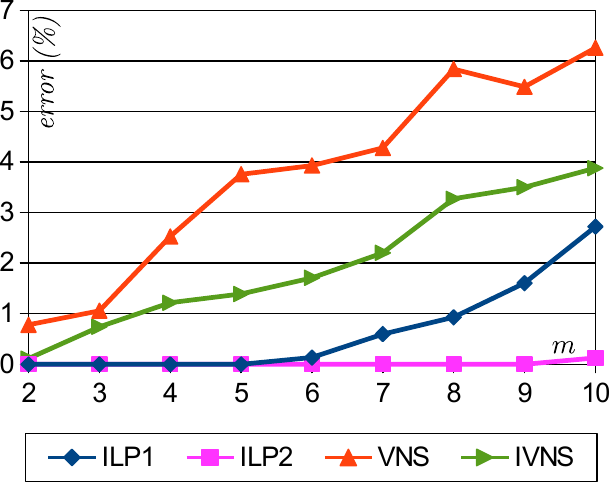}}
\put(9.5,0){b)}
\put(6.5,0.5){\includegraphics[height=4.2cm]{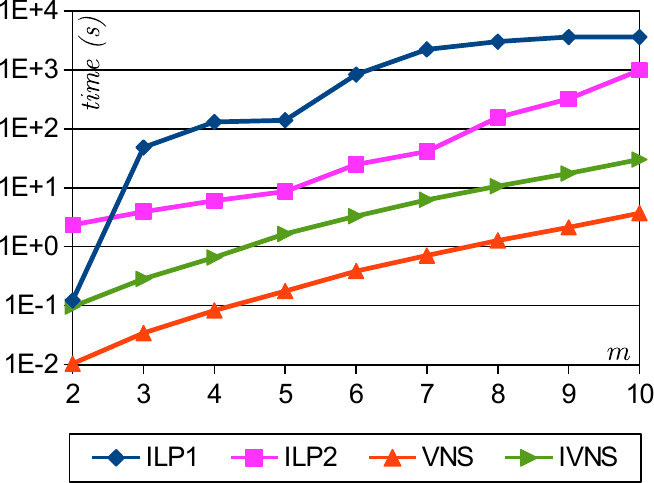}}
\end{picture}
\caption{\label{fig-vs-m-n-5m-delta-0.7}
Algorithm performance vs. $m$ for $n=5m$, $\delta=0.7$.
a) Average quality,
b) average execution time.
}
\end{figure}

The results obtained for tests with $m=2,\dots,10$, $n=5m$ and $\delta=0.7$ are shown in Fig. \ref{fig-vs-m-n-5m-delta-0.7}.
ILP1 found optimal solutions for all instances with $m\leq 5$ but was unsuccessful at some tests for each larger $m$.
In particular, it was not able to finish computations within an hour for any instances with $m\geq 9$.
The average ILP1 error is below 1\% for $m\leq 8$ and reaches 2.72\% for $m=10$.
ILP2 performed much better, finding the optimum schedules for all instances with $m\leq 9$.
Even for $m=10$, the average error of ILP2 is only 0.13\%.
The distance between VNS and IVNS schedules is small for $m\in\{2, 3\}$, but increases for larger $m$.
The errors obtained by VNS are below 6.5\%, and the IVNS errors are smaller than 4\%.
Thus, instances with $\delta=0.7$ are in general more demanding than those with $\delta=0.3$.
In the group of tests with $\delta=0.7$, the running time of ILP1 is close to that of IVNS only for $m=2$.
For $m\geq 3$, the running time of ILP1 significantly increases, and ILP1 becomes the slowest of all analysed algorithms.
In particular, for $m\geq 9$, the average running time of ILP1 is one hour, since it solved no instances within the imposed time limit.
In contrast, the average execution time of ILP2 is 985 seconds for $m=10$.

\begin{figure}[t]
\unitlength1cm
\begin{picture}(11,5)
\put(2.7,0){a)}
\put(0,0.5){\includegraphics[height=4.2cm]{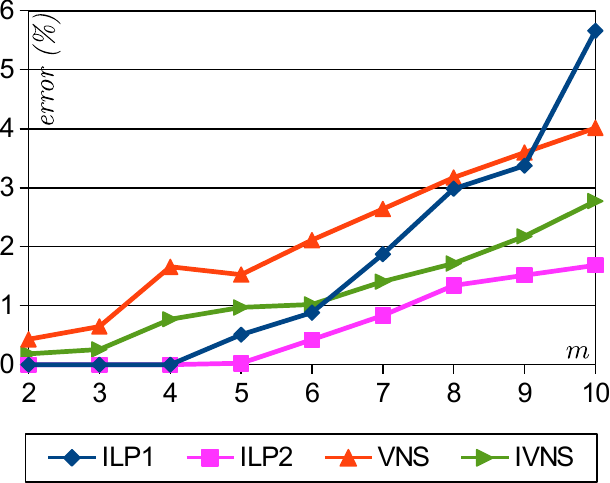}}
\put(9.5,0){b)}
\put(6.5,0.5){\includegraphics[height=4.2cm]{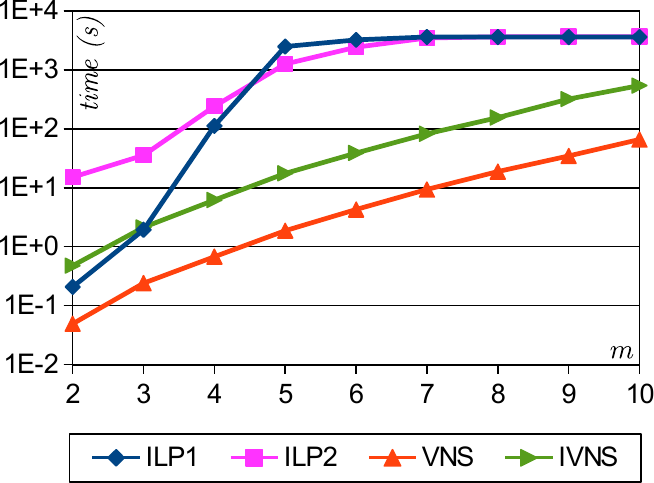}}
\end{picture}
\caption{\label{fig-vs-m-n-10m-delta-0.7}
Algorithm performance vs. $m$ for $n=10m$, $\delta=0.7$.
a) Average quality,
b) average execution time.
}
\end{figure}

The results obtained for the largest instances, with variable $m$, $n=10m$ and $\delta=0.7$, are shown in Fig. \ref{fig-vs-m-n-10m-delta-0.7}.
In this group, no tests with $m\geq 7$ were solved within the time limit by ILP1, and no tests with $m\geq 8$ were solved within the time limit by ILP2.
Hence, the running times of both these algorithms stabilise at one hour for large $m$.
There are significant differences between the qualities of solutions obtained by ILP1 and ILP2.
For $m\leq 4$, both these algorithms always find optimum solutions, but for larger $m$, ILP2 clearly outperforms ILP1.
For $m=10$, the average distance from $UB$ is 1.68\% in the case of ILP2 solutions, and 5.66\% in the case of ILP1 schedules. 
Moreover, ILP1 is outperformed by both VNS and IVNS for such $m$.
Thus, ILP1 is not recommended for solving large instances.
It is interesting that VNS and IVNS obtain here better results than for the tests with $n=5m$ and $\delta=0.7$.
The average errors of VNS reach 4.01\% for $m=10$ and the average distance of
IVNS schedules from the upper bound is below 3\% for all $m$.

In addition to the above experiments, a preliminary comparison of the methods in this paper with the ones in communications \cite{tresoldi21} and \cite{ha21preprint} is presented. The communications \cite{tresoldi21} and \cite{ha21preprint} appeared after this paper has been completed.
The algorithms ILP2 and IVNS are compared with the integer linear programming formulation F3 from \cite{ha21preprint}, and the integer linear program FT and heuristic HE from \cite{tresoldi21}.
The comparison is preliminary because of the use of different computers and the possible differences between Gurobi and CPLEX.

Algorithms F3, ILP2 and IVNS were run on the set of 432 instances used in \cite{ha21preprint}, obtained from the authors.
Algorithm ILP2 solved optimally 430 instances, and its average error was 0.02\%.
Formulation F3 delivered optimal solutions for 394 tests, and had an average error of 0.58\%.
The average execution time was 149 seconds in the case of ILP2 and 489 seconds in the case of F3.
Thus, it seems that ILP2 is more efficient than F3.
Heuristic IVNS obtained 339 optimal solutions, average error 0.70\% and average execution time 1.42 seconds.

Furthermore, algorithms ILP2 and IVNS were executed on the set of 3840 large instances used in \cite{tresoldi21}, which are publicly available.
Algorithm ILP2 solved all of them optimally in the average time of 1.76 seconds, while
FT is reported in \cite{tresoldi21} to solve all but one instance optimally, in the average time of 35 seconds.
These results cannot be compared directly because different computers and different solvers were used to run both algorithms.
However, the difference between the execution times seems significant.
Algorithm IVNS solved optimally 2430 instances, and HE found optimum solutions for 1851 tests \cite{tresoldi21}.
The average error of HE on instances not solved optimally is reported to be 1.91\%, which gives the average error of 0.99\% overall, while the average error of IVNS was 0.56\%.
The average computation time of HE is 0.21 seconds according to \cite{tresoldi21}, and the average running time of IVNS was 4.77 seconds.
Again, the time results cannot be compared directly.
Still, it seems that IVNS obtains better solutions than HE, but at a higher computational cost.

\section{Conclusions}

For the problem of maximising the total weight of on-time jobs with a common deadline, scheduled on parallel machines subject to a conflict graph, the paper presents a polynomial-time algorithm for the case of two parallel machines, and a performance guarantee which is tight for a broad family of algorithms for an arbitrary number of machines. Both these results were obtained for UET jobs. For jobs with general processing times, the paper presents two new integer linear programming formulations and a variable neighborhood search algorithm, which is embedded in an iterated search framework. Computational experiments show that both proposed integer linear programs obtain better results than those in \cite{ha2019capacitated}, ILP2 is particularly efficient for large instances, and the heuristic algorithm IVNS obtains good results in a reasonable time. Further development of such optimisation procedures, and their experimental evaluation, should be one of the main directions of future research. In particular, the experiments should include a more detailed comparison of the proposed algorithms with the methods recently announced in \cite{ha21preprint,tresoldi21}.

\bibliographystyle{splncs04}
\bibliography{ref}

\end{document}